\documentclass{article}
\usepackage{amsfonts}

\usepackage{graphicx}
\usepackage{amsmath}

\newtheorem{theorem}{Theorem}

\newtheorem{lemma}[theorem]{Lemma}

\newtheorem{remark}[theorem]{Remark}

\newenvironment{proof}[1][Proof]{\textbf{#1.} }{\ \rule{0.5em}{0.5em}}
\input{tcilatex}

\begin{document}

\title{The Gough-James Theory of Quantum Feedback Networks in the Belavkin
Representation}
\author{O. G. Smolyanov$^{1}$, A. Truman$^{2}$}
\date{}
\maketitle

\begin{abstract}
The mathematical theory of quantum feedback networks has recently been
developed by Gough and James \cite{QFN1} for general open quantum dynamical systems interacting
with bosonic input fields. In this article we show, that their feedback reduction formula for the coefficients of the closed-loop
quantum stochastic differential equation can be formulated in terms of Belavkin matrices. We show that the reduction formula leads to a non-commutative Mobius transformation based on Belavkin matrices, and establish a $\star$-unitary version of the Siegel identities.
\end{abstract}

1) Department of Mechanics and Mathematics, Moscow State University,
Vorob'evy Gory, Moscow, 119899 Russia.

2) Department of Mathematics, Swansea University, Swansea, Wales UK.

\section{Introduction}

In a recent publication \cite{QFN1} Gough and James have introduced a model
for a quantum feedback network. Each component may be modelled in isolation
as a Hudson-Parthasarathy model, valid for quantum optical models, and
represented as a single vertex with an equal number of external inputs and
outputs carried along semi-infinite transmission lines represented as
directed edges. The algebraic procedure is to collect all the operator
coefficients of the associated quantum stochastic differential equation
governing all components into a matrix. This gives the open-loop
description, and feedback is introduced by connecting various input/output
lines to giver internal edges. They use a Hamiltonian model for the entire
network which generalizes the Chebotarev-Gregoratti Hamiltonian which
describes the propagation of the fields along the edges and their
interaction at the vertices. They obtain a Markovian limit for the network
in a zero time delay limit, eliminating the internal edges in the process.
The limit quantum stochastic differential equation is then described by a
reduced operator matix.

We show that their feedback reduction formula is a M\"{o}bius transformation
that can naturally be extended to mappings into the Belavkin matrix
representation of quantum stochastic calculus. In particular we interpret
this as a non-commutative M\"{o}bius transformation between $\star $%
-unitaries. We begin by recalling the basic notions of quantum stochastic
calculus and its Belavkin formulation.

\subsection{Quantum Stochastic Processes}

The Hudson-Parthasarathy theory of quantum stochastic\ calculus considers
quantum stochastic processes as operator valued processes on Hilbert spaces
of the form $\frak{H}=\frak{h}_{0}\otimes \Gamma \left( \frak{k}\otimes
L^{2}[0,\infty )\right) $ where $\frak{h}_{0}$ is a fixed Hilbert space,
called the \emph{initial space}, and $\frak{k}$ is a fixed Hilbert space
called the \emph{internal space}. We shall be interested in the finite
dimensional case $\frak{k}=\mathbb{C}^{n}$ where $n\geq 1$. Here $\Gamma
\left( t\right) $ denotes the second quantization functor to (Bosonic) Fock
space. We shall denote the time variable $t$ by $A_{t}^{00}$. Taking $%
\left\{ \left| e_{i}\right\rangle :i=1,\cdots ,n\right\} $ to be an
orthonormal basis for $\frak{k}$, the creation process to state $\left|
e_{i}\right\rangle $ will be denoted as $A_{t}^{i0}$, while its adjoint, the
annihilator for the state, is denoted as $A_{t}^{0i}$. The scattering
process from state $\left| e_{j}\right\rangle $ to state $\left|
e_{i}\right\rangle $ will be denoted as $A_{t}^{ij}$. In this way, we have
the $\left( 1+n\right) \times \left( 1+n\right) $ \emph{fundamental quantum
processes} $A_{t}^{\alpha \beta }$. (We adopt the convention that Latin
indices range over $1,\cdots ,n$ while Greek indices range over $0,1,\cdots
,n$. We also apply a summation convention for repeated indices over the
corresponding ranges.) We note that we have $\left( A_{t}^{\alpha \beta
}\right) ^{\dag }=A_{t}^{\beta \alpha }.$

As is well known $\frak{H}$ decomposes as $\frak{H}_{\left[ 0,t\right]
}\otimes \frak{H}_{\left( t,\infty \right) }$ for each $t>0$ where $\frak{H}%
_{\left[ 0,t\right] }=\frak{h}_{0}\otimes \Gamma \left( \frak{k}\otimes
L^{2}[0,t)\right) $ and $\frak{H}_{\left( t,\infty \right) }=\Gamma \left( 
\frak{k}\otimes L^{2}(t,\infty )\right) $. We shall write $\frak{A}_{t]}$
for the space of operators on $\frak{H}$ that act trivially on the future
component $\frak{H}_{\left( t,\infty \right) }$. A quantum stochastic
process $X_{t}=\left\{ X_{t}:t\geq 0\right\} $ is said to be \emph{adapted}
if $X_{t}\in \frak{A}_{t]}$ for each $t\geq 0$.

Taking $\left\{ x_{\alpha \beta }\left( t\right) :t\geq 0\right\} $ to be a
family of adapted quantum stochastic processes, we may then form their
quantum stochastic integral $X_{t}=\int_{0}^{t}x_{\alpha \beta }\left(
s\right) dA_{s}^{\alpha \beta }$ where the differentials are understood in
the It\={o} sense. Given a similar quantum It\={o} integral $Y_{t}$, with $%
dY_{t}=y_{\alpha \beta }\left( t\right) dA_{t}^{\alpha \beta }$, we have the
quantum It\={o} product rule 
\begin{equation}
d\left( X_{t}.Y_{t}\right) =dX_{t}.Y_{t}+X_{t}.dY_{t}+dX_{t}.dY_{t},
\label{Ito formula}
\end{equation}
with the It\={o} correction given by 
\begin{equation}
dX_{t}.dY_{t}=x_{\alpha k}\left( t\right) y_{k\beta }\left( t\right)
dA_{t}^{\alpha \beta }.  \label{Ito correction}
\end{equation}

The coefficients $\left\{ x_{\alpha \beta }\left( t\right) \right\} $ may be
assembled into a matrix 
\begin{equation}
\mathbf{X}_{t}=\left( 
\begin{tabular}{l|l}
$x_{00}\left( t\right) $ & $x_{0\bullet }\left( t\right) $ \\ \hline
$x_{\bullet 0}\left( t\right) $ & $x_{\bullet \bullet }\left( t\right) $%
\end{tabular}
\right) \in \frak{A}_{t]}^{\left( 1+n\right) \times \left( 1+n\right) },
\label{Ito matrix}
\end{equation}
which we call the \emph{It\={o} matrix} for the process. (Here we use the
convention that $x_{0\bullet }\left( t\right) $\ denotes the row vector with
entries $\left( x_{0j}\left( t\right) \right) _{j=1}^{n}$, etc. The It\={o}
matrix for a product $X_{t}Y_{t}$ of quantum It\={o} integrals will then
have entries $\left\{ x_{\alpha \beta }\left( t\right) Y_{t}+X_{t}y_{\alpha
\beta }\left( t\right) +x_{\alpha k}\left( t\right) y_{k\beta }\left(
t\right) \right\} $ and is therefore given by $\mathbf{X}_{t}Y_{t}+X_{t}%
\mathbf{Y}_{t}+\mathbf{X}_{t}\mathbf{PY}_{t}$, where $\mathbf{P}:=\left( 
\begin{tabular}{l|l}
$0$ & $0$ \\ \hline
$0$ & $\mathtt{I}_{n}$%
\end{tabular}
\right) $.

\subsection{Belavkin's Matrix Representation}

We consider the mapping from It\={o} matrices $\mathbf{X}\in \frak{A}%
^{\left( 1+n\right) \times \left( 1+n\right) }$ to associated \textit{%
Belavkin matrices} 
\begin{equation}
\mathbb{X}=\left( 
\begin{tabular}{l|l|l}
$0$ & $x_{0\bullet }$ & $x_{00}$ \\ \hline
$0$ & $x_{\bullet \bullet }$ & $x_{\bullet 0}$ \\ \hline
$0$ & $0$ & $0$%
\end{tabular}
\right) \in \frak{A}^{\left( 1+n+1\right) \times \left( 1+n+1\right) }.
\label{Belavkin matrix}
\end{equation}
We also introduce 
\begin{equation*}
\begin{array}{cc}
\mathbb{I}_{n}:=\left( 
\begin{tabular}{l|l|l}
$1$ & $0$ & $0$ \\ \hline
$0$ & $\mathtt{I}_{n}$ & $0$ \\ \hline
$0$ & $0$ & $1$%
\end{tabular}
\right) , & \mathbb{J}_{n}:=\left( 
\begin{tabular}{l|l|l}
$0$ & $0$ & $1$ \\ \hline
$0$ & $\mathtt{I}_{n}$ & $0$ \\ \hline
$1$ & $0$ & $0$%
\end{tabular}
\right) ,
\end{array}
\end{equation*}
where $\mathtt{I}_{n}$ is the $n\times n$ identity matrix. The subscripts $n$
will be generally dropped from now on for convenience. We have the following
identifications 
\begin{eqnarray*}
\mathbf{X}^{\dag } &\longleftrightarrow &\mathbb{X}^{\star }:=\mathbb{JX}%
^{\dag }\mathbb{J}\text{,} \\
\mathbf{XPY} &\longleftrightarrow &\mathbb{XY}\text{,} \\
\mathbf{XY} &\longleftrightarrow &\mathbb{XJY}\text{.}
\end{eqnarray*}

We shall refer to $\mathbb{X}^{\star }$ as the $\star -$\emph{involution} of 
$\mathbb{X}$. The It\={o} differential $dX_{t}=x_{\alpha \beta }\left(
t\right) dA_{t}^{\alpha \beta }$ may then be written as 
\begin{equation*}
dX_{t}=tr\left\{ \mathbb{X}_{t}d\widetilde{\mathbb{A}}_{t}\right\} ,
\end{equation*}
where (with $\prime $ denoting the usual transpose for arrays) 
\begin{equation*}
d\widetilde{\mathbb{A}}_{t}:=\left( 
\begin{tabular}{l|l|l}
$0$ & $0$ & $0$ \\ \hline
$\left( dA^{0\bullet }\right) ^{\prime }$ & $\left( dA^{\bullet \bullet
}\right) ^{\prime }$ & $0$ \\ \hline
$dA^{00}$ & $\left( dA^{\bullet 0}\right) ^{\prime }$ & $0$%
\end{tabular}
\right) .
\end{equation*}
The main advantage of using this representation is that the It\={o}
correction $\mathbf{XPY}$ ca now be given as just the ordinary product $%
\mathbb{XY}$\ of the Belavkin matrices.

Let $X_{t}$ and $Y_{t}$ be quantum stochastic integrals, then the quantum
It\={o} product rule may be written as 
\begin{equation}
d\left( X_{t}Y_{t}\right) =tr\left\{ \left[ \left( X_{t}\text{$\mathbb{I}$}+%
\mathbb{X}_{t}\right) \left( Y_{t}\text{$\mathbb{I}$}+\mathbb{Y}_{t}\right)
-\left( X_{t}Y_{t}\right) \text{$\mathbb{I}$}\right] d\widetilde{\mathbb{A}}%
_{t}\right\} .
\end{equation}
The process $f\left( X_{t}\right) $ has differential $df\left( X_{t}\right)
=tr\{[f\left( X_{t}\text{$\mathbb{I}$}+\mathbb{X}_{t}\right) -f\left(
X_{t}\right) \mathbb{I}]d\widetilde{\mathbb{A}}_{t}\}.$

\subsection{Evolutions and Dynamical Flows}

Hudson and Parthasarathy \cite{HP} show that the quantum stochastic
differential equation (QSDE) 
\begin{equation}
dV_{t}=tr\left\{ \mathbb{G}V_{t}d\widetilde{\mathbb{A}}_{t}\right\} ,\quad
V_{0}=1,
\end{equation}
has a unique solution for a given constant Belavkin matrix $\mathbb{G}=%
\mathbb{V}-\mathbb{I}$ of coefficients on $\frak{B}\left( \frak{h}%
_{0}\right) $, the bounded operators on $\frak{h}_{0}$. Necessary and
sufficient conditions for unitarity are then given in Belavkin
representation by 
\begin{equation*}
\left( \mathbb{I}+\mathbb{G}\right) \left( \mathbb{I}+\mathbb{G}\right)
^{\star }=\mathbb{I}=\left( \mathbb{I}+\mathbb{G}\right) ^{\star }\left( 
\mathbb{I}+\mathbb{G}\right) .
\end{equation*}
This states that $\mathbb{V}=\mathbb{I}+\mathbb{G}$ is $\star $-unitary on $%
\frak{B}\left( \frak{h}_{0}\right) ^{\left( 1+n+1\right) \times \left(
1+n+1\right) }$, that is 
\begin{equation}
\mathbb{VV}^{\star }=\mathbb{I}=\mathbb{V}^{\star }\mathbb{V}.
\end{equation}
We may write the QSDE as $dV_{t}=tr\left\{ \left[ \mathbb{V}V_{t}-\mathbb{I}%
V_{t}\right] d\widetilde{\mathbb{A}}_{t}\right\} $.

\begin{lemma}
The most general form for $\mathbb{V}$ leading to a unitary is 
\begin{equation}
\mathbb{V}=\left( 
\begin{tabular}{l|l|l}
$1$ & $-L^{\dag }S$ & $-\frac{1}{2}L^{\dag }L-iH$ \\ \hline
$0$ & $S$ & $L$ \\ \hline
$0$ & $0$ & $1$%
\end{tabular}
\right) \equiv \left( 
\begin{tabular}{l|l|l}
$V_{00}$ & $V_{0k}$ & $V_{00^{\prime }}$ \\ \hline
$V_{k0}$ & $V_{kk}$ & $V_{k0^{\prime }}$ \\ \hline
$V_{0^{\prime }0}$ & $V_{0^{\prime }k}$ & $V_{0^{\prime }0^{\prime }}$%
\end{tabular}
\right) .  \label{V}
\end{equation}
where $S$ is a unitary in $\frak{B}\left( \frak{h}_{0}\right) ^{n\times n}$, 
$L$ is a column vector length $n$ with entries in $\frak{B}\left( \frak{h}%
_{0}\right) $, and $H$ is self-adjoint in $\frak{B}\left( \frak{h}%
_{0}\right) $.
\end{lemma}

The proof follows from the analysis of \cite{HP}. The triple $\left(
S,L,H\right) $ is termed the Hudson-Parthasarathy parameters of the open
system evolution. In standard notation the QSDE reads as (summ over all
field multiplicities) 
\begin{equation*}
dV_{t}=\left\{ \left( S_{jk}-\delta _{jk}\right)
dA_{t}^{jk}+L_{j}dA_{t}^{j0}-L_{j}^{\ast }S_{jk}dA_{t}^{0j}-\left( \frac{1}{2%
}L_{j}^{\ast }L_{j}+iH\right) dA_{t}^{00}\right\} V_{t}.
\end{equation*}
We interpret $A_{t}^{\alpha \beta }$ as the input noise and $V_{t}^{\ast
}A_{t}^{\alpha \beta }V_{t}$ as the output noise.

\section{Quantum Cascaded Systems}

If two systems are cascaded in series then the Hudson-Parthasarathy
parameters of the composite system were shown to be \cite{GJ Series},\cite
{QFN1} 
\begin{eqnarray*}
S_{\text{series}} &=&S_{2}S_{1}, \\
L_{\text{series}} &=&L_{2}+S_{2}L_{1}, \\
H_{\text{series}} &=&H_{1}+H_{2}+\text{Im}\left\{ L_{2}^{\dag
}S_{2}L_{1}\right\} .
\end{eqnarray*}
Here the output of the first sytem $\left( S_{1},L_{1},H_{1}\right) $ is fed
forward as the input to the second system $\left( S_{2},L_{2},H_{2}\right) $
and the limit of zero time delay is assumed. As remarked in \cite{QFN1}, the
series product actually arises natural in Belavkin matrix form as 
\begin{equation*}
\mathbb{V}_{\text{series}}=\mathbb{V}_{2}\mathbb{V}_{1}.
\end{equation*}
The product is clearly associative, as one would expect physically, and the
general rule for several systems in series is then $\mathbb{V}_{\text{series}%
}=\mathbb{V}_{n}\cdots \mathbb{V}_{2}\mathbb{V}_{1}$

\section{General Feedback Reduction Formula}

The internal edges may be eliminated in a zero time delay limit to obtain a
reduced model. Let $0<n_{\mathtt{i}}<n$ be then number of internal edges to
be eliminated, and let $n_{\mathtt{e}}=n-n_{\mathtt{i}}$ be the remaining
edges. The algebraic information about the original network is contained in
the matrix $\mathbb{V}$ which we partition as 
\begin{equation*}
\mathbb{V}=\left( 
\begin{tabular}{l|l|l}
$V_{00}$ & $V_{0\mathtt{e}}\;V_{0\mathtt{i}}$ & $V_{00^{\prime }}$ \\ \hline
$
\begin{array}{c}
V_{\mathtt{e}0} \\ 
V_{\mathtt{i}0}
\end{array}
$ & $
\begin{array}{cc}
V_{\mathtt{ee}} & V_{\mathtt{ei}} \\ 
V_{\mathtt{ie}} & V_{\mathtt{ii}}
\end{array}
$ & $
\begin{array}{c}
V_{\mathtt{e}0^{\prime }} \\ 
V_{\mathtt{i}0^{\prime }}
\end{array}
$ \\ \hline
$V_{0^{\prime }0}$ & $V_{0^{\prime }\mathtt{e}}\;V_{0^{\prime }\mathtt{i}}$
& $V_{0^{\prime }0^{\prime }}$%
\end{tabular}
\right) .
\end{equation*}
Here we decompose indices into two groups $\mathtt{e}$ and $\mathtt{i}$
distinguishing external and internal. That is, $V_{00}=V_{0^{\prime
}0^{\prime }}=1$, $V_{\mathtt{e}0}=V_{\mathtt{i}0}=V_{0^{\prime
}0}=V_{0^{\prime }\mathtt{e}}=V_{0^{\prime }\mathtt{i}}=0$%
\begin{eqnarray*}
S &=&\left[ 
\begin{array}{cc}
V_{\mathtt{ee}} & V_{\mathtt{ei}} \\ 
V_{\mathtt{ie}} & V_{\mathtt{ii}}
\end{array}
\right] =\left[ 
\begin{array}{cc}
S_{\mathtt{ee}} & S_{\mathtt{ei}} \\ 
S_{\mathtt{ie}} & S_{\mathtt{ii}}
\end{array}
\right] , \\
L &=&\left[ 
\begin{array}{c}
V_{\mathtt{e}0^{\prime }} \\ 
V_{\mathtt{i}0^{\prime }}
\end{array}
\right] =\left[ 
\begin{array}{c}
L_{\mathtt{e}} \\ 
L_{\mathtt{i}}
\end{array}
\right] ,
\end{eqnarray*}
and $\left[ V_{0\mathtt{e}}\;V_{0\mathtt{i}}\right] =-SL^{\ast }$.

\begin{theorem}
We assemble a Belavkin matrix $\mathcal{F}\left( \mathbb{V},X\right) $ in $%
\frak{A}^{\left( 1+n_{\mathtt{e}}+1\right) \times \left( 1+n_{\mathtt{e}%
}+1\right) }$ with sub-blocks 
\begin{equation*}
\mathcal{F}\left( \mathbb{V},X\right) _{\alpha \beta }=V_{\alpha \beta
}+V_{\alpha \mathtt{i}}X\left( 1-V_{\mathtt{ii}}X\right) ^{-1}V_{\mathtt{i}%
\beta }
\end{equation*}
for $\alpha =0,\mathtt{e},0^{\prime }$ and $\beta =0,\mathtt{e},0^{\prime }$%
, where we fix a unitary operator in $X\in \mathbb{C}^{n_{\mathtt{i}}\times
n_{\mathtt{i}}}$ such that the inverse above exists. Then $\mathcal{F}\left( 
\mathbb{V},X\right) $ is again a $\star $-unitary, that is, 
\begin{equation}
\mathcal{F}\left( \mathbb{V},X\right) ^{\star }\mathcal{F}\left( \mathbb{V}%
,X\right) =\mathcal{F}\left( \mathbb{V},X\right) \mathcal{F}\left( \mathbb{V}%
,X\right) ^{\star }=\mathbb{I},
\end{equation}
so that $\mathcal{F}\left( \mathbb{V},X\right) $ determines a unitary
dynamics for the reduced set of $n_{\mathtt{e}}$ inputs. Moreover, we have
the identity 
\begin{equation}
\mathcal{F}\left( \mathbb{V},X\right) ^{\star }=\mathcal{F}\left( \mathbb{V}%
^{\star },X^{\dag }\right) .  \label{xxx}
\end{equation}
\end{theorem}

\begin{remark}
The matrix $X$ appearring above is typically an adjacency matrix in
applications, describing which internal outputs are to be connected to which
internal inputs. In engineering, it could be interpreted as a gain matrix.
We also point out that the involutions in $\left( \ref{xxx}\right) $\ are on
spaces of different dimensions. The first involves $\mathbb{J}_{n_{\mathtt{i}%
}}$ while the second involves $\mathbb{J}_{n}$.
\end{remark}

\begin{proof}
The construction of $\mathcal{F}\left( \mathbb{V},X\right) $\ is essentially
the rephrasing of the M\"{o}bius transformation associated with the
reduction, introduced in \cite{QFN1},\ in the language of Belavkin matrices.
The construction in $\left( \ref{V}\right) $ clearly yields a Belavkin
matrix over the remaining $n_{\mathtt{e}}$ external degrees of freedom. It
is a straightforward calculation to show that, for $X$ unitary, the matrix
takes the form 
\begin{equation*}
\mathcal{F}\left( \mathbb{V},X\right) =\left( 
\begin{tabular}{l|l|l}
$1$ & $-L^{\text{red}\dag }S$ & $-\frac{1}{2}L^{\text{red}\dag }L^{\text{red}%
}-iH^{\text{red}}$ \\ \hline
$0$ & $S^{\text{red}}$ & $L^{\text{red}}$ \\ \hline
$0$ & $0$ & $1$%
\end{tabular}
\right) ,
\end{equation*}
with the Hudson-Parthasarathy parametrizing operators $\left( S^{\text{red}%
},L^{\text{red}},H^{\text{red}}\right) $\ given by 
\begin{eqnarray*}
S^{\text{red}} &=&S_{\mathtt{ee}}+S_{\mathtt{ei}}\left( X^{-1}-S_{\mathtt{ii}%
}\right) ^{-1}S_{\mathtt{ie}}, \\
L^{\text{red}} &=&L_{\mathtt{e}}+S_{\mathtt{ei}}^{-1}\left( X^{-1}-S_{%
\mathtt{ii}}\right) L_{\mathtt{i}}, \\
H^{\text{red}} &=&H+\sum_{i=\mathtt{i},\mathtt{e}}\text{Im}L_{j}^{\dag }S_{j%
\mathtt{i}}^{-1}\left( X^{-1}-S_{\mathtt{ii}}\right) L_{\mathtt{i}},
\end{eqnarray*}
in agreement with \cite{QFN1}.
\end{proof}

For $\mathbb{X}\in \frak{A}^{\left( 1+n+1\right) \times \left( 1+m+1\right)
} $ we shall introduce the extended convention $\mathbb{X}^{\star }:=\mathbb{%
J}_{m}\mathbb{X}^{\dag }\mathbb{J}_{n}$. Let $\mathbb{V}$ be the Belavkin
matrix generating a unitary quantum dynamics as above, we define the
M\"{o}bius transformation $\Phi :D\mapsto \frak{A}^{\left( 1+n_{\mathtt{e}%
}+1\right) \times \left( 1+n_{\mathtt{e}}+1\right) }$ by $\Phi =\mathcal{F}%
\left( \mathbb{V},\cdot \right) $ with domain $D=\{X\in \mathbb{C}^{n_{%
\mathtt{i}}\times n_{\mathtt{i}}}:\mathbb{I}-V_{\mathtt{ii}}X$ \ is
invertible\}.

\begin{theorem}
The mapping $\Phi $ satisfies the Siegel type identities 
\begin{eqnarray*}
\Phi \left( X\right) ^{\star }\Phi \left( Y\right) &=&\mathbb{I}+\left( 
\begin{array}{c}
V_{0\mathtt{i}} \\ 
V_{\mathtt{ei}} \\ 
V_{0^{\prime }\mathtt{i}}
\end{array}
\right) ^{\star }\left( 1-X^{\dag }V_{\mathtt{ii}}^{\dag }\right)
^{-1}\left( X^{\dag }Y-1\right) \left( 1-V_{\mathtt{ii}}Y\right) ^{-1}\left( 
\begin{array}{c}
V_{0\mathtt{i}} \\ 
V_{\mathtt{ei}} \\ 
V_{0^{\prime }\mathtt{i}}
\end{array}
\right) , \\
\Phi \left( X\right) \Phi \left( Y\right) ^{\star } &=&\mathbb{I}+\left( V_{%
\mathtt{i}0},V_{\mathtt{ie}},V_{\mathtt{i}0^{\prime }}\right) \left( 1-XV_{%
\mathtt{ii}}\right) ^{-1}\left( XY^{\dag }-1\right) \left( 1-V_{\mathtt{ii}%
}^{\dag }Y\right) ^{-1}\left( V_{\mathtt{i}0},V_{\mathtt{ie}},V_{\mathtt{i}%
0^{\prime }}\right) ^{\star }.
\end{eqnarray*}
In particular, $\Phi $ maps unitaries to $\star $-unitaries.
\end{theorem}

\begin{proof}
The form of these relations are similar to the standard Siegel identities,
see for instance \cite{Young}, but with the $\star $-involution now
replacing the usual \dag . The algebraic manipulations involved are
otherwise identical.
\end{proof}

We remark that the standard Siegel type identities have independently been
extended in an entirely different direction to deal with Bogoliubov
transformations in a recent paper of Gough, James and Nurdin \cite{SQLFN}.
They replace the usual \dag -involution with an alternative involution, this
time on the space of doubled up matrices required to describe the symplectic
structure, however they similarly rely on the argument used in the above
proof.


\bigskip

\begin{thebibliography}{99}
\bibitem{Gardiner}  C. Gardiner and P. Zoller, Quantum Noise: A Handbook of
Markovian and Non-Markovian Quantum Stochastic Methods with Applications to
Quantum Optics, 2nd ed., ser. Springer Series in Synergetics. Springer, 2000.

\bibitem{Wiseman}  H.Wiseman, \textit{Quantum theory of continuous feedback}%
, Phys. Rev. A, vol. 49, no. 3, pp. 2133.2150, 1994.

\bibitem{YK1}  M. Yanagisawa, H. Kimura, \textit{Transfer function approach
to quantum control Part I: Dynamics of quantum feedback systems}, IEEE
Transactions on Automatic Control, \textbf{48}, No. 12, 2107-2120, December
(2003)

\bibitem{YK2}  M. Yanagisawa, H. Kimura, \textit{Transfer function approach
to quantum control Part II: Control concepts and applications}, IEEE
Transactions on Automatic Control, \textbf{48}, No. 12, 2121-2132, December
(2003)

\bibitem{QFN1}  J. Gough, M.R. James, \textit{Quantum\ Feedback Networks:
Hamiltonian Formulation}, Commun. Math. Phys., 287, 1109-1132 (2009)

\bibitem{GJ Series}  J. Gough, M.R. James, \textit{The series product and
its application to feedforward and feedback networks}, IEEE Trans. Automatic
Control, 54(11):25302544 (2009)

\bibitem{JNP}  M. R. James, H. I. Nurdin, and I. R. Petersen, \textit{H}$%
^{\infty }$\textit{\ control of linear quantum stochastic systems}, 2007, to
be published (accepted 24-9-2007) in IEEE Transactions on Automatic Control.
[Online]. Available: http://arxiv.org/pdf/quant-ph/0703150(pre-print)

\bibitem{Lloyd}  S. Lloyd, \textit{Coherent quantum control}, Phys. Rev. A,
62:022108, 2000

\bibitem{Mabuchi}  H. Mabuchi, \textit{Coherent-feedback control with a
dynamic compensator}, March 2008, submitted for publication, preprint:
http://arxiv.org/abs/0803.2007.

\bibitem{NJP}  H. I. Nurdin, M. R. James, and I. R. Petersen, \textit{%
Quantum LQG control with quantum mechanical controllers}, 2008, to be
presented at the 17th IFAC World Congress (Seoul, South Korea, July 6-11,
2008). [Online]. Available: http://arxiv.org/pdf/0711.2551(expandedversion)

\bibitem{HP}  R. L. Hudson and K. R. Parthasarathy, \textit{Quantum Ito's
formula and stochastic evolutions,} Commun. Math. Phys. \textbf{93}, 301-323
(1984)

\bibitem{partha}  K. Parthasarathy, An Introduction to Quantum Stochastic
Calculus. Berlin: Birkhauser, 1992.

\bibitem{G Wong-Zakai}  J. Gough, \textit{Quantum Stratonovich calculus and
the quantum Wong-Zakai theorem}, J. Math. Phys., vol. 47, no. 113509, 2006.

\bibitem{Young}  N. Young, An Introduction to Hilbert Space, Cambridge
Mathematical Textbooks, (1988)

\bibitem{SQLFN}  J.E. Gough, M.R. James, and H.I. Nurdin, \textit{Squeezing
components in linear quantum feedback networks}, Phys. Rev. A 81, 023804
(2010)
\end{thebibliography}
\end{document}